\pdfoutput=1
\documentclass{article}
\usepackage{ijcai20}

\usepackage{times}
\usepackage{soul}
\usepackage{url}
\usepackage[hidelinks]{hyperref}
\usepackage[small]{caption}
\usepackage{graphicx}
\usepackage{amsmath}
\usepackage{amsthm}
\usepackage{amssymb}
\usepackage{booktabs}
\usepackage[ruled,vlined,linesnumbered]{algorithm2e}
\usepackage{lipsum}
\usepackage[inline]{enumitem}

\usepackage{color,soul}
\usepackage{xspace}

\usepackage{xargs}

\newenvironment{myintro}%
  {\list{}{\leftmargin=0.1in\rightmargin=0.1in}\item[]}%
  {\endlist}

\newtheorem{theorem}{Theorem}
\newtheorem{corollary}{Corollary}
\newtheorem{conjecture}{Conjecture}
\newtheorem{lemma}{Lemma}
\newtheorem{definition}{Definition}
\newtheorem{proposition}{Proposition}

\newcommand{\nop}[1]{}
\newcommand{\fhw}{\mathit{fhw}}
\newcommand{\semantic}[1]{\ensuremath{\mathtt{sem}\mbox{-}{#1}}}
\newcommand{\sghw}{\ensuremath{\semantic{ghw}}}
\newcommand{\sfhw}{\ensuremath{\semantic{\fhw}}}

\newcommand{\ssubw}{\ensuremath{\semantic{subw}}}

\newcommand{\np}{\textsc{NP}\xspace}

\newcommand{\dpcomplexity}{\textsc{DP}\xspace}
\newcommand{\ptime}{\textsc{PTIME}\xspace}

\newcommand{\nxn}{n\times n}
\newcommand{\core}{\mathit{core}}
\newcommand{\csp}{\textsc{CSP}\xspace}
\newcommand{\pcsp}{p\textsc{-CSP}\xspace}
\newcommand{\bucq}{p\textsc{-UCQ}\xspace}

\newcommand{\classH}{\mathcal{H}\xspace}
\newcommand{\classA}{\mathcal{A}\xspace}
\newcommand{\classU}{\mathcal{U}\xspace}
\newcommand{\classC}{\mathcal{C}\xspace}
\newcommand{\classGrid}{\mathcal{G}_{\nxn}\xspace}
\newcommand{\degen}{exotic\xspace}

\newcommand{\bfA}{\ensuremath{\mathbf{A}}}

\newcommand{\bfC}{\ensuremath{\mathbf{C}}}
\newcommand{\bfB}{\ensuremath{\mathbf{B}}}
\newcommand{\bfD}{\ensuremath{\mathbf{D}}}

\usepackage{tabularx, environ, caption}
\makeatletter
\newcolumntype{\expand}{}
\long\@namedef{NC@rewrite@\string\expand}{\expandafter\NC@find}

\NewEnviron{problem}[2][]{%
  \def\problem@arg{#1}%
  \def\problem@framed{framed}%
  \def\problem@lined{lined}%
  \def\problem@doublelined{doublelined}%
  \ifx\problem@arg\@empty%
    \def\problem@hline{}%
  \else%
    \ifx\problem@arg\problem@doublelined%
      \def\problem@hline{\hline\hline}%
    \else%
      \def\problem@hline{\hline}%
    \fi%
  \fi%
  \ifx\problem@arg\problem@framed%
    \def\problem@tablelayout{|>{\itshape}lX|c}%
    \def\problem@title{\multicolumn{2}{|l|}{%
        \raisebox{-\fboxsep}{\textsc{#2}}%
      }}%
  \else
    \def\problem@tablelayout{>{\itshape}lXc}%
    \def\problem@title{\multicolumn{2}{l}{%
        \raisebox{-\fboxsep}{\textsc{\Large #2}}%
      }}%
  \fi%
  \smallskip\par\noindent%
  \begin{tabularx}{\columnwidth}{\expand\problem@tablelayout}%
    \problem@hline%
    \problem@title\\[2\fboxsep]%
    \BODY\\\problem@hline%
  \end{tabularx}%
  \medskip\par%
}
\makeatother

\title{Semantic Width and the Fixed-Parameter Tractability of \\Constraint Satisfaction Problems}

\author{
Hubie Chen$^1$ \and
Georg Gottlob$^{2,3}$ \and
Matthias Lanzinger$^{3}$ \And
Reinhard Pichler$^3$\\
\affiliations
$^1$Birkbeck, University of London\\
$^2$Oxford University\\
$^3$TU Wien\\
\emails
h.chen@dcs.bbk.ac.uk,
georg.gottlob@cs.ox.ac.uk,
\{mlanzing, pichler\}@dbai.tuwien.ac.at%
}

\begin{document}

\maketitle

\begin{abstract}

  Constraint satisfaction problems (CSPs) are an important formal framework for the uniform treatment of various prominent AI tasks, e.g., coloring or scheduling problems. Solving CSPs is, in general, known to be \np-complete and fixed-parameter intractable when parameterized by their constraint scopes.  We give a characterization of those classes of CSPs for which the problem becomes fixed-parameter tractable.
  Our characterization significantly increases the utility of the CSP framework by making it possible to decide the fixed-parameter tractability of problems via their CSP formulations.
  We further extend our characterization to the evaluation of unions of conjunctive queries, a fundamental problem in databases. Furthermore, we provide some new insight on the frontier of \ptime solvability of CSPs.
  In particular, we observe that bounded fractional hypertree width is more general than bounded hypertree width only for classes that exhibit a certain type of exponential growth.
  The presented work resolves a long-standing open problem and yields powerful new tools for complexity research in AI and database theory.

  \noindent
  \emph{This is an extended version of~\cite{DBLP:conf/ijcai/ChenGLP20}.}
\end{abstract}

\section{Introduction}

CSPs are a fundamental problem of artificial intelligence. As a unifying formal framework, they play a foundational role in many areas of AI research, see e.g.,~\cite{DBLP:journals/aim/Kumar92,DBLP:conf/aaai/Narvaez18}.
However, the unifying aspect of CSPs has not yet reached its full potential. While a CSP formulation of a problem allows for reuse of common algorithmic strategies and implementations~\cite{DBLP:journals/ai/GottlobLS00,DBLP:journals/ai/DoK01}, results in computational complexity still often require individual investigation, with little help from the framework. A complexity characterization for CSP would allow researchers to finally leverage the CSP framework also for strong computational complexity results, hence greatly simplifying the study of all problems that can be formulated as CSPs. The consequences and wide-reaching applications of such a characterization motivate our central research question. Note that throughout this paper, the parameterized complexity of CSPs always refers to the problem parameterized by the size of its constraint scopes.

\begin{myintro}
   \noindent
   \textbf{Research Challenge:} Is there a natural characterization of
   the fixed-parameter tractable classes of CSPs?
\end{myintro}

To be precise, we study what is referred to as the uniform CSP problem in the literature. In the uniform problem, we are interested in how the structure of constraint scopes affects the complexity of the problem, i.e., we characterize restrictions to the structure of constraint scopes. In the nonuniform problem, one considers restrictions to the constraint relations. Here, in a classic result, \cite{DBLP:journals/jct/HellN90} gave an elegant characterization of \ptime solvability. More recently,~\cite{DBLP:conf/focs/Bulatov17} and~\cite{DBLP:conf/focs/Zhuk17} were able to independently establish a powerful dichotomy theorem. However, results for the nonuniform case do not translate to the uniform problem.

There is a long line of research devoted to the (parameterized) computational complexity of solving CSPs based on structural parameters of their associated hypergraphs.
In a landmark result, \cite{DBLP:journals/jacm/Grohe07} resolved the question for a restricted class of CSPs; namely those with bounded arity.  There, \ptime decidability is fully characterized by bounded treewidth modulo homomorphic equivalence. Moreover, for bounded arity, we have fixed-parameter tractability if and only if the problem is solvable in \ptime.

To tackle the problem beyond bounded arity, a number of generalizations of treewidth have been developed that provide sufficient conditions for tractably solving CSPs, the most important of which are hypertree width~\cite{DBLP:journals/jcss/GottlobLS02} and fractional hypertree width~\cite{DBLP:journals/talg/GroheM14}. Yet, bounding these parameters yields only sufficient conditions for tractability. A necessary condition for unbounded arity remains elusive.  In the parameterized space, a highly impressive result by~\cite{DBLP:journals/jacm/Marx13} was able to characterize those hypergraphs, i.e., problem structures, that always allow for fixed-parameter tractable evaluation by the \emph{submodular width} of the hypergraphs. However, while this result is closely related to our goal, the fact that the characterization is on the hypergraph level significantly limits its applicability in our setting: When we consider the CSP formulation of a problem, then the complexity of our problem does not depend on the complexity of other, unrelated, CSPs that happen to have the same underlying hypergraphs. Hence, characterizing on the hypergraph level restricts us to a worst-case that may not be connected to the problem we want to study (this point is discussed in detail in Section~\ref{sec:gap}).

Despite their unquestionable importance, Grohe's and Marx's characterizations do not answer our research question. 
 Instead, we introduce a new parameter -- \emph{semantic submodular width} ($\ssubw$) -- to capture the minimal submodular width over the (infinite) equivalence class of semantically equivalent CSPs. We show that it is still possible to decide $\ssubw$ and find the minimal semantically equivalent CSP in time that depends only on the size of the parameter. Following that, we give a reduction from Marx's setting to ours, which allows us to prove the necessary lower bound. Akin to Marx's characterization, our result assumes the Exponential Time Hypothesis~\cite{DBLP:journals/jcss/ImpagliazzoPZ01}; a standard assumption of parameterized complexity.

\begin{myintro}
  \noindent
  \textbf{Main Result 1:}  Assuming the Exponential Time Hypothesis, a class of CSPs is fixed-parameter tractable if and only if it has bounded semantic submodular width.
\end{myintro}

Through the well-known equivalence of CSP to the homomorphism problem as well as conjunctive query containment~\cite{DBLP:journals/jcss/KolaitisV00} and evaluation~\cite{DBLP:books/cs/Maier83}, our main result also applies to those important problem families. By adapting our notion of $\ssubw$ from CSPs to the more general notion of \emph{unions of conjunctive queries} (UCQs) accordingly, we can also extend our characterization result to UCQs, an important and widely studied class of query languages in database theory~\cite{DBLP:journals/jacm/SagivY80,DBLP:journals/jacm/AtseriasDK06}.

\begin{myintro}
  \noindent
  \textbf{Main Result 2:}  Assuming the Exponential Time Hypothesis, a class of UCQs is fixed-parameter tractable if and only if it has bounded semantic submodular width.
\end{myintro}

With the question of fixed-parameter tractability resolved, we shift our attention to \ptime solvable classes of CSPs. Here, a characterization of tractable restrictions for the uniform CSP problem remains an open question. We briefly discuss how our parameterized results relate to the non-parameterized case. 
Furthermore, we utilize some recent results on the connection of hypergraph width parameters and Vapnik-Chervonenkis dimension to derive new insight on the frontier of tractability of the uniform CSP problem. In particular, we show that the two most important sufficient conditions for tractable CSP solving -- bounded fractional hypertree width and bounded hypertree width -- actually collapses for classes of CSPs as long as they do not exhibit a certain kind of, seemingly unnatural, exponential growth.

The rest of the paper is structured as follows. Section~\ref{sec:prelim} recalls necessary definitions for constraint satisfaction problems, unions of conjunctive queries, and relevant hypergraph width parameters. We expand on the differences to Marx's characterization in Section~\ref{sec:gap} before we present our two main results in Sections~\ref{sec:csp} and~\ref{sec:ucqs}. Section~\ref{sec:ptime} presents some new insights regarding the \ptime solvability of CSPs.
We end with concluding remarks in Section~\ref{sec:conclusion}. Moreover, we include an appendix that includes full proofs of all statements that are not already shown in the main body of text. In Appendix~\ref{sec:shw} we show that semantic hypertree width behaves differently than the other widths investigated in our setting and prove a characterization in terms of semantic generalized hypertree width.

\section{Preliminaries}
\label{sec:prelim}

\subsection{Parameterized Complexity}
\label{sec:prelimcomplexity}

Parameterized complexity enables a more fine-grained study of
computational complexity. Here, we give an abridged definition of the
notions necessary for this paper. For full definitions
and details we refer to~\cite{DBLP:series/txtcs/FlumG06}.

For an alphabet of symbols $\Sigma$, a
\emph{parameterized problem} is given as a pair $(P, \kappa)$ of a
problem $P \subseteq \Sigma^*$ and its parameterization $\kappa$ that
maps each string in $\Sigma^*$ to a parameter.

We say that a parameterized problem $(P, \kappa)$ is
\emph{fixed-parameter tractable} if there exists an algorithm that
decides whether a given string $x \in \Sigma^*$ is in $P$ in time
$f(\kappa(x))poly(|x|)$, where $f$ is a computable function and $poly$
is a polynomial. 

Let $(P, \kappa)$ and $(P', \kappa')$ be two parameterized problems.
A \emph{fpt-reduction} from $(P, \kappa)$ to $(P', \kappa')$ is a mapping
$R : \Sigma^* \to \Sigma^*$ with the following properties:
\begin{enumerate}[topsep=0pt,noitemsep,label=(\arabic*)]
\item $x \in P \iff R(x) \in P'$ for every $x \in \Sigma^*$,
\item $R$ is computable in time $f(\kappa(x))poly(|x|)$ ($f$ is computable), and
\item there is a computable function $g$ such that $\kappa'(x) \leq g(\kappa(x))$ for all $x \in \Sigma^*$.
\end{enumerate}
We say $(P, \kappa)$ is fpt-reducible to $(P', \kappa')$, denoted $(P, \kappa) \leq (P',\kappa')$.
The class of fixed-parameter tractable problems is closed under fpt-reductions.

Our main results assume the Exponential Time Hypothesis, which states
that 3-SAT with $n$ variables can not be decided in $2^{o(n)}$
time~\cite{DBLP:journals/jcss/ImpagliazzoPZ01}. This is a standard
assumption of parameterized complexity theory.

\subsection{Constraint Satisfaction Problems}
\label{sec:prelimcsps}

We formalize CSPs as a relational homomorphism problem. A \emph{signature} is a finite set of relation symbols with associated arities. A \emph{(relational) structure} $\bfA$ (over signature $\sigma$) consists of a domain $A$ and an interpretation $R^\bfA$ for each relation symbol $R$ in the signature. 
  Let $\bfA, \bfB$ be relational structures, we call a function $h: A \to B$
a \emph{homomorphism} from $\bfA$ into $\bfB$, if for every relation symbol $R$ and all $(x_1, \dots, x_{ar(R)}) \in R^{\bfA}$ also
$(h(x_1), \dots, h(x_{ar(R)})) \in R^{\bfB}$, where $ar(R)$ is the arity of $R$. We write $|\bfA| = |\sigma|+|A|+\sum_{R\in \sigma}|R^\bfA| ar(R)$ for the size of structure $\bfA$.

We call an ordered pair $(\bfA, \bfB)$ of structures a \emph{constraint satisfaction problem instance}. Intuitively, $\bfA$ expresses the constraint scopes and $\bfB$ the permitted assignments for each constraint scope. For a class $\classA$ of structures, the corresponding \emph{constraint satisfaction problem}, denoted $\csp(\classA)$, is the following decision problem.
\begin{problem}[framed]{$\csp(\classA)$}
  Instance: & A CSP instance $(\bfA,\bfB)$
  where $\bfA \in \mathcal{A}$. \\
  Question: & Is there a homomorphism from $\bfA$ into $\bfB$?
\end{problem}
By slight abuse of notation, we also call $\classA$ a class of \emph{constraint satisfaction problems}.
Note that what we call $\csp(\classA)$ is sometimes denoted as $\csp(\classA, -)$ to emphasize that we are dealing with the uniform CSP problem (cf.,~\cite{DBLP:journals/jacm/Grohe07}). Furthermore, constants play no role in our considerations since they can be eliminated by straightforward preprocessing.

A \emph{hypergraph} $H$ is a tuple $(V(H), E(H))$, where $V(H)$ is the set of vertices and $E(H)\subseteq 2^{V(H)}$ the set of \emph{hyperedges}. 
For a set $U \subseteq V(H)$, we define the subhypergraph induced by $U$ as $H[U] = (U, E')$ where $E' = \{e \cap U \mid e \in E(H)\} \setminus \{\emptyset\}$.
The hypergraph $H(\bfA)$ of a structure $\bfA$ is the hypergraph where the vertices equal $A$ and $e \in E(H(\bfA))$ if and only if there exists some relation symbol $R$ such that some permutation of $e$ is contained in $R^\bfA$. The hypergraph of a CSP instance $(\bfA, \bfB)$ is the hypergraph of $\bfA$, i.e., the hypergraph of a CSP instance represents only the structure of its constraint scopes. 
We are interested in how this structure affects the complexity of the decision problem. We thus consider the $\csp$ decision problem parameterized by its constraint scope structure:
\begin{problem}[framed]{$\pcsp(\classA)$}
  Instance: & A CSP instance $(\bfA,\bfB)$
  where $\bfA \in \mathcal{A}$. \\
  Parameter: & $|\bfA|$ \\
  Question: & Is there a homomorphism from $\bfA$ into $\bfB$?
\end{problem}

For a class $\classH$ of hypergraphs, let $Struct[\classH]$ denote all structures whose hypergraphs are in $\classH$. We will abbreviate the problem $\csp(Struct[\classH])$ to $\csp(\classH)$, i.e., \csp restricted to those instances whose hypergraphs are in $\classH$. The analogue applies to $\pcsp$.

For two structures $\bfA$ and $\bfA'$, we say $\bfA$ is
\emph{homomorphically equivalent} to $\bfA'$, or $\bfA \simeq \bfA'$,
if there exists a homomorphism from $\bfA$ into $\bfA'$ and vice versa. The \emph{core} of a structure $\bfA$, denoted $\core(\bfA)$,  is the minimal structure (with regards to the number of tuples) that is homomorphically equivalent to $\bfA$. It is not hard to verify that every structure has a unique (up to isomorphism) core. For a class of structures $\classA$, we write $\core(\classA)$ for the class of cores of structures in $\classA$.

In the context of CSPs, we say a structure $\bfA$ is \emph{contained} in structure $\bfA'$ if for every $\bfB$ we have
that if $(\bfA, \bfB)$ has a solution, then $(\bfA', \bfB)$ also has a solution. It is easy to see that $\bfA$ is contained in $\bfA'$ if and only if there exists a homomorphism from $\bfA'$ to $\bfA$. If two structures $\bfA$ and $\bfA'$ are contained within each other, we say that they are \emph{semantically equivalent} (we write $\bfA \equiv \bfA'$). Hence, if $\bfA \equiv \bfA'$ then  
for every $\bfB$ we have that  $(\bfA, \bfB)$ has a solution if and only if $(\bfA', \bfB)$ has a solution.
Furthermore, note that $\bfA \equiv \bfA'$ if and only if $\bfA \simeq \bfA'$, i.e., homomorphic equivalence equals semantic equivalence. In particular, $(\bfA, \bfB)$ is always equivalent to $(core(\bfA), \bfB)$.

\subsection{Unions of Conjunctive Queries}
\label{sec:prelimucq}

Please note that, for consistency and brevity, we will define unions of conjunctive queries via CSPs. This does not match the standard presentations of the problem but is equivalent to them.

An instance of the \emph{(boolean) unions of conjunctive queries} (UCQ) problem is a set of structures $\{\bfA_1, \dots, \bfA_n\}$, we write $\bigcup_{i=1}^n\bfA_i$, and a structure $\bfB$ which is usually referred to as the database.
We say an instance of the UCQ problem $(\bigcup_{i=1}^n\bfA_i, \bfB)$ has a solution if any of the CSP instances $(\bfA_i, \bfB)$, for $1 \leq i \leq n$, has a solution. Hence, the accompanying parameterized decision problem for a class of UCQs $\mathcal{U}$ is the following
 
\begin{problem}[framed]{$\bucq(\mathcal{U})$}
  Instance: & A UCQ $U = \bigcup_{i=1}^n \bfA_i$ where $U \in \mathcal{U}$ and a database $\bfB$. \\
  Parameter: & $\sum_{i=1}^n |\bfA_i|$ \\ 
  Question: & Does $U, \bfB$ have a solution?
\end{problem}

Analogue to CSPs, the equivalence of UCQs will be important.
We say that two UCQs $U=\bigcup_{i=1}^n \bfA_i$ and $U'=\bigcup_{i=1}^m \bfA'_i$ are \emph{semantically equivalent} (we write $U \equiv U'$)
if for every structure $\bfB$, $(U, \bfB)$ has a solution if and only if $(U', \bfB)$ has a solution.

A UCQ $\bigcup_{i=1}^n \bfA_i$ is \emph{non-redundant} if there are no $\bfA_i$ and $\bfA_j$ ($i\neq j$) such that $\bfA_i$ is contained in $\bfA_j$. Note that every UCQ can be transformed into an equivalent non-redundant UCQ by repeated deletion of structures that are contained by other structure in the UCQ~\cite{DBLP:journals/jacm/SagivY80}). We write $nr(U)$ for the UCQ obtained by applying this procedure to make an UCQ $U$ non-redundant. Importantly, as the procedure only deletes structures we have $nr(U) \subseteq U$.
\subsection{Decompositions and Their Widths}
\label{sec:widths}
In this work we will only consider width notions that are based on tree decompositions.
A tuple $(T, (B_u)_{u \in T})$ is a \emph{tree
  decomposition} of a hypergraph $H$ if $T$ is a tree, every $B_u$
 is a subset of $V(H)$ and the following two conditions are satisfied:
\begin{enumerate}[topsep=0pt, noitemsep,label=(\arabic*)]
\item 
 For every $e \in E(H)$ there
 is a node $u \in T$ s.t. $e \subseteq B_u$, and
\item for every vertex $v \in V(H)$,
  $\{u \in T \mid v \in B_u\}$ is connected in $T$.
\end{enumerate}
For functions $f\colon 2^{V(H)} \to \mathbb{R}^+$, the
\emph{$f$-width} of a tree decomposition is
$\sup\{f(B_u) \mid u \in T\}$ and the $f$-width of a hypergraph is the
minimal $f$-width over all its tree decompositions.  Let $\mathcal{F}$
be a class of functions from subsets of $V(H)$ to the non-negative
reals, then the $\mathcal{F}$-width of $H$ is
$\sup \{f\mbox{-width}(H)\mid f\in \mathcal{F} \}$. All such
widths are implicitly extended to structures and CSP instances by taking the
width of their respective hypergraphs.

The following properties of functions $f\colon 2^{V(H)}\to \mathbb{R}^+$ are important: 
\begin{itemize}[topsep=0pt,noitemsep]
\item $f$ is \emph{monotone} if $X \subseteq Y$ implies $f(X) \leq f(Y)$.
\item $f$ is called \emph{edge-dominated} if $f(e) \leq 1$ for every $e \in E(H)$. 
\item $f$ is called \emph{submodular} if $f(X) + f(Y) \geq f(X \cap Y) + f(X \cup Y)$ holds
for every $X, Y \subseteq V(H)$. 
\end{itemize}
We say a weight function $\gamma: E(H) \to \mathbb{R}^+$ is a \emph{fractional edge cover}
of a set $X \subseteq V(H)$ if for every $v \in X$ we have $\sum_{e \in I_v} \gamma(e) \geq 1$ where $I_v$
is the set of all edges incident to $v$. If we restrict the co-domain to $\{0,1\}$, we obtain the definition of an \emph{integral edge cover}. We refer to the total weight  $\sum_{e \in E(H)} \gamma(e)$ of an edge cover as
the \emph{size} of the edge cover.

For $X \subseteq V(H)$, let $\rho_H(X)$ be the
size of the smallest integral edge cover of $X$ by edges in $E(H)$ and
$\rho_H^*(X)$ the size of the smallest fractional edge cover of $X$ by edges in $E(H)$.
This framework now allows us to define many of the important widths in the current literature.
\begin{description}[noitemsep,topsep=0pt,leftmargin=!]
\item[(Primal) Treewidth of $H$~\cite{DBLP:journals/jal/RobertsonS86}:]
  $tw(H) := c$-width, where $c(X) = |X|-1$.
\item[Generalized hypertree width of $H$~\cite{DBLP:journals/jcss/GottlobLS02}:]  
  $ghw(H):=\rho_H$-width.
\item[Fractional hypertree width of $H$~\cite{DBLP:journals/talg/GroheM14}:]
  $\fhw(H) := \rho_H^*$-width.
\item[Submodular width of $H$~\cite{DBLP:journals/jacm/Marx13}:]
  $subw(H) := \mathcal{F}$-width$(H)$,
  where $\mathcal{F}$ is the set of all monotone, edge-dominated,
  submodular functions 
  $b$ on $2^{V(H)}$ with $b(\emptyset)=0$.    
\end{description}
A notable omission, that is not expressible through this notion of
$f$-width, is \emph{hypertree width (hw)}~\cite{DBLP:journals/jcss/GottlobLS02}, which uses the same width function
as $ghw$ but imposes an additional restriction on the
tree decomposition. Details of hypertree width are not important for the main part of this paper. We formally introduce them in Appendix~\ref{sec:shw} where we present some novel results on the behaviour of hypertree width in homomorphically equivalent structures.
Note that these widths spawn a hierarchy in the sense that the following inequality holds for all hypergraphs $H$:
\[
  subw(H) \leq \fhw(H) \leq ghw(H) \leq hw(H) \leq tw(H)+1
\]
For a class of structures $\classA$, we say $\classA$ has bounded width if there exists a constant $k$ such that every structure in $\classA$ has width $\leq k$.
The computational complexity of CSP is tightly linked
to this hierarchy of parameters. This connection is summarized by the
following two propositions.

\begin{proposition}[\cite{DBLP:journals/talg/GroheM14}]
  \label{prop:fhw}
  Let $\classC$ be a class of CSP instances of bounded $\fhw$. Then $\csp(\classC)$ is tractable.
\end{proposition}

\begin{proposition}[\cite{DBLP:journals/jacm/Marx13}]
  \label{prop:subw}
  Let $\classH$ be a recursively enumerable class of hypergraphs. Assuming the Exponential Time Hypothesis, $\pcsp(\classH)$ is fixed-parameter tractable if and only if $\classH$ has bounded submodular width.
\end{proposition}

\section{Main Results}
\label{sec:main}

\subsection{Characterization of Hypergraph Classes vs. Classes of CSP Instances}
\label{sec:gap}

Recall the motivation given in the introduction. Many AI problems have natural CSP formulations and
we wish to determine the computational complexity of all such problems through a characterization of the complexity of \csp.
In this section we argue why a characterization on the hypergraph level (which ignores relation symbols), as in Proposition~\ref{prop:subw}, is not enough for this goal.
The main issue with the hypergraph characterization is that even though a CSP instance may have a highly complex hypergraph structure, it can still be easy to solve. Yet, the complexity of $\pcsp(\classH)$ expresses only the complexity of the worst-case CSP instances of the given structure. We illustrate this issue in the following example.

Consider the following problem: Given a directed graph $G$, can we embed (by a homomorphism) a bidirected $\nxn$-grid into $G$? The corresponding CSP instance $C_n(G)=(\bfA, \bfB)$ has a single relation symbol $E$ and $\bfB = G$. As domain of $\bfA$ we take $\{x_{i,j} \mid i\in [n], j\in [n]\}$ and $E^\bfA$ contains exactly the following tuples specifying the $\nxn$-grid: $(x_{i,j}, x_{i+1,j}), (x_{i+1,j}, x_{i,j})$ for $i \in [n-1], j\in[n]$ and $(x_{i,j}, x_{i,j+1})$, $(x_{i,j+1}, x_{i,j})$ for $i \in [n], j\in[n-1]$.

We now consider the class $\classC$ of all CSP instances $C_n(G)$ for
$n \geq 1$ and all graphs $G$. The hypergraphs of $\classC$ are, by
definition, exactly the class of $n\times n$-grid graphs $\classGrid$,
which is well-known to have unbounded treewidth~\cite{DBLP:journals/jal/RobertsonS86}. In general, it is
difficult to determine the submodular width of graphs since the
definition depends on a supremum over an infinite class of
functions. However, Lemma~\ref{lem:twsubw} below provides us with a
convenient way to recognize that certain classes have unbounded
submodular width.

\begin{lemma}
  \label{lem:twsubw}
  Let $H$ be an arbitrary hypergraph and let $rank(H)$ be the maximum edge size in $H$, then
  \[tw(H) \leq rank(H) \cdot subw(H)\]
\end{lemma}
\begin{proof}[Sketch]
  Let $f: X \mapsto |X|/rank(H)$ be a function on the subsets of
  $V(H)$. It is easy to verify that $f$ is submodular, edge-dominated
  and monotone. For any node $u$ of any tree decomposition of $H$ we
  clearly have $|B_u| = rank(H) \cdot f(B_u)$ and therefore also
  $tw(H)+1 = rank(H) \cdot f$-width$(H)$.  Since $f$ is submodular,
  edge-dominated and monotone we also have $f$-width$(H) \leq subw(H)$
  and the statement follows immediately.
\end{proof}

From Lemma~\ref{lem:twsubw} we can conclude that $\classGrid$ also has unbounded submodular width.
From Proposition~\ref{prop:subw} we can thus only deduce that $\pcsp(\classGrid)$ is not fixed-parameter tractable.

However, for every $C_n(G) = (\bfA, \bfB)$, we have that $\core(\bfA)$ is the structure with domain $\{x_1, x_2\}$ and $E^{\core(\bfA)} = \{(x_1, x_2), (x_2, x_1)\}$. This is easy to verify, e.g., by observing that an undirected $\nxn$-grid is 2-colorable. Clearly, $(\core(\bfA), \bfB)$ is solvable in polynomial time and it is equivalent to $(\bfA, \bfB)$. It follows that $\pcsp(\classC)$ is in fact fixed-parameter tractable (and indeed tractable), despite the complexity of $\pcsp(\classGrid)$. We see that a hypergraph level characterization has inherent shortcomings in  establishing lower bounds for specific problem classes.

\subsection{Constraint Satisfaction Problems}
\label{sec:csp}

In this section we prove our characterization theorem for CSPs. 
The discussion in Section~\ref{sec:gap} shows that unbounded submodular width can still allow for fixed-parameter tractable \csp solving. Hence, we require a new, more general, property to fully capture fixed-parameter tractability. We follow~\cite{DBLP:journals/sigmod/BarceloPR17}
who introduced the notion of \emph{semantic generalized hypertree width} and define the following general notion of \emph{semantic widths} of CSPs. 

\begin{definition} Let $\classA$ be the class of all structures and $w: \classA \to \mathbb{R}^+$ be invariant under isomorphism. We define \emph{semantic $w$} as
  $\semantic{w}(\bfA) := \inf \{ w(\bfA') \mid \bfA' \equiv \bfA \}$.
\end{definition}

Using this definition, we are now ready to
state our first main result.
We show that the characterization from Proposition~\ref{prop:subw} can indeed be strengthened to the following characterization of the fixed-parameter tractability of CSP instances.

\begin{theorem}
  \label{thm:main}
  Let $\classA$ be a recursively enumerable class of CSPs. Assuming the Exponential Time Hypothesis, $\pcsp(\classA)$ is fixed-parameter tractable if and only if $\classA$ has bounded semantic submodular width.
\end{theorem}

Our proof of the theorem relies on two central lemmas. First, we show how bounded semantic submodular width leads to fixed-parameter tractability. The basic idea is simple, instead of solving a CSP instance with possibly arbitrarily high submodular width, we want to solve an equivalent instance with low width. However, it is not clear how to find such an equivalent instance and whether finding it is decidable.
For generalized hypertree width~\cite{DBLP:journals/sigmod/BarceloPR17} have recently shown, that for any structure $\bfA$, $\sghw(\bfA)$ is precisely $ghw(\core(\bfA))$.
Indeed, we show in Lemma~\ref{lem:swidth}, that the same connection also holds for the more complex cases of fractional hypertree width and submodular width. Note that for treewidth this property is trivial since treewidth is hereditary, i.e., removing edges from a hypergraph can not increase its treewidth. The width functions considered here are not hereditary and involve additional technical considerations beyond those necessary for the $ghw$ case.

\begin{lemma}
  \label{lem:swidth}
  For every structure $\bfA$:
  \begin{enumerate}
  \item $\semantic{\rho^*}(\bfA) = \rho^*(\core(\bfA))$ \label{swidth:rho}
  \item $\sfhw(\bfA) = \fhw(\core(\bfA))$
  \item $\ssubw(\bfA) = subw(\core(\bfA))$
  \end{enumerate}
\end{lemma}
\begin{proof}[Proof (Sketch)]
  First, since call equivalent structures have isomorphic cores it is
  enough to show $w(core(\bfA)) \leq w(\bfA)$ to establish that
  $\semantic{w}(\bfA) = w(core(\bfA))$ for any invariant $w$.

  Let $H$ be the hypergraph of $\bfA$ and $H'$ the hypergraph of
  $core(\bfA)$. Note that there exists an homomorphism $h$ from $\bfA$
  to $core(\bfA)$ where $h(a) = a$ for all elements in the domain of
  $core(\bfA)$. From this we can then show that for every tree decomposition
  $(T, (B_u)_{u \in T})$ of $H$, there exists a tree decomposition $(T, (B'_u)_{u \in T})$
  of $H'$ where $B'_u = B_u \cap V(H')$.

  For the $\fhw$ case we then use an observation on how fractional edge
  covers behave under homomorphisms to show that this transformation
  does not increase the $\rho^*$-width. Hence, we can transform the tree decomposition for $H$ with minimal $\rho^*$-width
  into a new tree decomposition for $H'$ with less or equal $\rho^*$-width, i.e., $\fhw(H') \leq \fhw(H)$.
  The observation for edge covers under homomorphisms also leads
  to the result for $\semantic{\rho^*}$.

  The $subw$ case requires additional considerations as the width is
  now defined over a whole class of functions $\mathcal{F}$.  We show
  that for every edge-dominated, submodular function $f'$ over $H'$
  there exists an edge-dominated submodular function $f$ over $H$
  such that $f'$-width$(H') \leq f$-width$(H)$. In particular, for every $f'$
  this the function $f: X \mapsto f'(X \cap V(H'))$ has the required properties.
\end{proof}
While $\semantic{\rho^*}$ and $\sfhw$ are less general than $\ssubw$
they will be of further interest in the discussion of \ptime solvability of CSPs in Section~\ref{sec:ptime}.
In the
context of our main result, the most important consequence of
Lemma~\ref{lem:swidth} is that we are always able to find the
equivalent structure with minimal submodular width by simply computing
the core.
In principle, finding the core of a structure is intractable (formally, deciding if a structure $\bfA'$ is the core of a structure $\bfA$, is \dpcomplexity-complete~\cite{DBLP:journals/tods/FaginKP05}). However, in our parameterized setting
the computation of the core of $\bfA$ only depends on the parameter.

To establish a lower bound for classes with unbounded semantic submodular width we will make use of previous results from~\cite{DBLP:journals/toct/ChenM15}. A step in our reduction will require an additional definition that helps us fix the domains of individual elements in the reduction.
For a structure $\bfA$, let $\bfA^*$ be the expansion of $\bfA$ by a new fresh unary relation symbol $U_a$ with $U_a^\bfA = \{a\}$ for every element of the domain $a \in A$. For a class of structures $\classA$ we write $\classA^*$ for $\{\bfA^* \mid \bfA \in \classA \}$. Our intention is to establish our lower bound by reduction from the hypergraph setting of Proposition~\ref{prop:subw}. We will make use of the following two reductions.

\begin{proposition}[\cite{DBLP:journals/toct/ChenM15}]
  \label{prop:star}
  Let $\classA$ be a recursively enumerable class of structures. Then
  \[
    \pcsp(\core(\classA)^*) \leq \pcsp(\classA)
  \]
\end{proposition}

\begin{lemma}
  \label{lem:redh}
  Let $\classA$ be a recursively enumerable class of structures and let $\classH^{\classA}$ be the class of hypergraphs of $\classA$.
  \[
    \pcsp(\classH^{\classA}) \leq \pcsp(\classA^*)
    \]
\end{lemma}
\begin{proof}
  Let $(\bfC,\bfD)$ be an instance of $\csp(\classH^{\classA})$ and let $H$ be the hypergraph of $\bfC$ and $D$ be the domain of $\bfD$. Recall that edges can represent multiple constraint scopes, i.e., multiple tuples in $\bfC$. For each edge $e \in E(H)$, we consider the sets $F_{t_1},\dots,F_{t_k}$ of satisfying assignments $e \to D$ for each of the tuples $t_1,\dots,t_k$ of $\bfC$ that become edge $e$ in the hypergraph. We then produce the set $F_e= \bigcap_{i=1}^k F_{t_i}$ of satisfying assignments over all the tuples for $e$. Observe that computing $F_e$ for all $e \in E(H)$ is possible in polynomial time.
   
By definition there exists a structure $\bfA^*$ in $\classA^*$ where
$\bfA$ has hypergraph $H$. We can compute such a $\bfA^*$ by
enumeration of $\classA$ until we find an $\bfA$ with a matching
hypergraph and then computing $\bfA^*$ from $\bfA$.

We will reduce
$(\bfC,\bfD)$ to $(\bfA^*, \bfB)$ where $\bfB$ is constructed as
follows. As the domain of $\bfB$ we take $A \times D$. For each $a \in A$ we have a $U_a^{\bfA}$ with
$U_a^{\bfA} = \{a\}$. Let $U_a^\bfB = \{ (a,d) \mid d \in D\}$. For
each other relation symbol $R$ of $\bfA$ and each tuple
$(a_1, \dots, a_k) \in R^\bfA$, we add tuples
$((a_1,f(a_1)),\dots, (a_k,f(a_k)))$ to $R^\bfB$ where $f \in F_e$ and
$e$ is the hyperedge $\{a_1, \dots, a_k\}$.

We now show that $(\bfC, \bfD)$ has a solution iff $(\bfA^*, \bfB)$
has a solution. First, 
suppose $h$ is a homomorphism from $\bfC$ to
$\bfD$ and note that $\bfA^*$ and $\bfC$ have the same domain since
$\bfA$ and $\bfC$ have the same underlying hypergraph.
It is then not difficult to see that $g: a \mapsto (a, h(a))$ is a homomorphism from $\bfA^*$ to $\bfB$:
For the unary relations $U_a^\bfA$, the image trivially exists in $U_a^\bfB$.
For the other relations, it is enough to observe hat for every edge $e$ of $H$, the assignment $h$ restricted to variables in $e$ must be in $F_e$.

For the other side, observe that a homomorphism $g$ from $\bfA^*$ to $\bfB$ must be of the form $a \mapsto (a, h(a))$. We argue that $h$ is a homomorphism from $\bfC$ to $\bfD$.
As $\bfA$ and $\bfC$ have the same domain, $h$ also applies to the domain of $\bfC$. By definition of $F_e$
we have that for every tuple $\bar{x}$ in $R^\bfC$, $h$ maps to a tuple in $R^\bfD$ as long as $\bar{x}$ is covered by some edge of $\bfA$. Since the hypergraphs are the same, this holds for all the tuples in $\bfC$ and therefore $h$ is a homomorphism.
\end{proof}

\emph{Proof of Theorem~\ref{thm:main}.}
Let $\classH^{\core(\classA)}$ be the class of hypergraphs of the structures in $\core(\classA)$. We claim that the two problems $\pcsp(\classA)$ and $\pcsp(\classH^{\core(\classA)})$ are fpt-reducible to each other. If the claim holds, $\pcsp(\classA)$ is fixed-parameter tractable iff $\pcsp(\classH^{\core(\classA)})$ is fixed-parameter tractable. By Proposition~\ref{prop:subw} this is the case iff $\classH^{\core(\classA)}$ has bounded submodular width. By Lemma~\ref{lem:swidth}, this is equivalent to $\classA$ having bounded semantic submodular width.

What is left, is to show the claim. First, we observe:
\[
  \pcsp(\classA) \leq \pcsp(\core(\classA)) \leq \pcsp(\classH^{\core(\classA)})
\]
The left reduction holds because $(\bfA, \bfB)$ is equivalent to $(core(\bfA), \bfB)$ and computing the core is feasible in $f(|\bfA|)$ time. The right reduction is trivial since all instances of $\pcsp(\core(\classA))$ are also instances of $\pcsp(\classH^{\core(\classA)})$. For the other direction we get the intended reduction by straightforward combination of Lemma~\ref{lem:redh} and Proposition~\ref{prop:star}:
\[
   \pcsp(\classH^{\core(\classA)}) \leq \pcsp(\core(\classA)^*) \leq \pcsp(\classA)
 \]
\qed

\subsection{Unions of Conjunctive Queries}
\label{sec:ucqs}

We now extend the characterization in Theorem~\ref{thm:main} from
CSPs to UCQs. To do so we first need to introduce a way to extend the
relevant definitions to UCQs. For our width notions the natural
extension to UCQs is through the maximum of its parts, i.e., for width
function $w$ and UCQ $U = \bigcup_{i=1}^n \bfA_i$ let
$w(U) := \max\{w(\bfA_i)\mid i\in [n]\}$.  Semantic width functions
are defined the same as for CSPs, i.e.,
$\semantic{w} := \inf \{ w(U') \mid U' \equiv U \}$. However,
equivalence of UCQs is more complex than equivalence in CSPs. In particular, the characterization by homomorphic equivalence is no longer applicable. Therefore,
some additional effort is required to determine the analogue of
Lemma~\ref{lem:swidth}.  Using the
following classic result by Sagiv and Yannakakis we can derive the
fitting Lemma~\ref{lem:ucqssubw}.

\begin{proposition}[\cite{DBLP:journals/jacm/SagivY80}]
  \label{prop:nrucq}
  Let $U = \bigcup_{i=1}^n \bfA_i$ and $U' = \bigcup_{j=1}^m \bfA'_j$ be non-redundant UCQs.
  Then $U \equiv U'$ if and only if for every $\bfA_i$ there is a unique $\bfA'_j$ such that $\bfA_i \equiv \bfA'_j$.
\end{proposition}

\begin{lemma}
  \label{lem:ucqssubw}
  Let $U$ be an UCQ, then
  \[
    \ssubw(U) = \max\{ subw(core(\bfA_i)) \mid \bfA_i \in nr(U) \}
  \]
\end{lemma}
\begin{proof}
  It is clear that the right side of the equality is the $subw$ of an
  UCQ that is equivalent to $U$. All that is to show is that this is
  in fact the minimal subw of an equivalent UCQ. %
  For the sake of brevity we will write core-$subw(nr(U))$ for $\max\{ subw(core(\bfA_i)) \mid \bfA_i \in nr(U) \}$ in the rest of the argument.
  
  Proof is by contradiction. Suppose there exist a UCQ $V \equiv U$ with $subw(V) < $ core-$subw(nr(U))$. Since $V \equiv U$, clearly also $nr(V) \equiv nr(U)$. Furthermore, since $nr(V) \subseteq V$ (recall the construction of $nr(V)$) we also have $subw(nr(V)) < $ core-$subw(nr(U)$. Now, from Proposition~\ref{prop:nrucq} we have that for every $\bfB_i \in nr(V)$, there is an equivalent $\bfA_j\in nr(U)$. By Lemma~\ref{lem:swidth} it follows that
  $subw(\bfB_i) \geq \ssubw(\bfB_i) = \ssubw(\bfA_j) = subw(core(\bfA_j))$
  for all such combinations of $\bfB_i$ and $\bfA_j$. From the
  definition of $subw$ for UCQs this then gives an immediate contradiction of
  $subw(nr(V)) < $ core-$subw(nr(U))$.
\end{proof}

From Lemma~\ref{lem:ucqssubw} it is now easy to see, that for a class of UCQs $\classU$ with bounded $\ssubw$, the $\bucq(\classU)$ problem is fixed-parameter tractable. For every $U$ in $\classU$ we can simply compute $nr(U) = \bigcup_{i=1}^n \bfA_i$ and then solve the CSPs $(core(\bfA_i), \bfB)$ individually. In combination with Theorem~\ref{thm:main} we see that this procedure is fixed-parameter tractable.

To establish the lower bound, we make use of previous work on the complexity of existential positive logic~\cite{DBLP:journals/tocl/Chen14}. The result there is stated in a different setting but a translation is not difficult through the well-known equivalence of solving CSPs and model checking of primitive positive first-order formulas.

\begin{proposition}[Theorem 3.2 in \cite{DBLP:journals/tocl/Chen14}]
  \label{prop:ucqred}
  Let $\classU$ be recursively enumerable class of non-redundant UCQs and let $\classA$ be the class of all individual structures that make up the UCQs in $\classU$. Then $\pcsp(\classA) \leq \bucq(\classU)$.
\end{proposition}

\begin{theorem}
  Let $\classU$ be a recursively enumerable class of UCQs. Assuming the Exponential Time Hypothesis, $\bucq(\classU)$ is fixed-parameter tractable if and only if\, $\classU$ has bounded semantic submodular width. 
\end{theorem}
\begin{proof}
  For the case where $\classU$ has bounded semantic submodular width
  we have already given a fixed-parameter tractable procedure for
  solving $\bucq(\classU)$ above. We will establish the lower bound by
  introducing the class $nr(\classU) = \{ nr(U) \mid U \in \classU\}$
  as an intermediate.
  
  Suppose $\classU$ has unbounded $\ssubw$ and let $\classA$ be the
  class of all individual structures that make up the UCQs in $nr(\classU)$.
  From Lemma~\ref{lem:ucqssubw} it follows that $nr(\classU)$ and $\classA$ both also have unbounded $\ssubw$.
  By Proposition~\ref{prop:ucqred} we have $\pcsp(\classA) \leq \bucq(nr(\classU))$ and therefore, by Theorem~\ref{thm:main},
  $\bucq(nr(\classU))$ can not be fixed-parameter tractable.

  To finish the proof we show that
  $\bucq(nr(\classU)) \leq \bucq(\classU)$. The reduction is
  straightforward, an instance $(U, \bfB)$ of $\bucq(nr(\classU))$ is
  reduced to the instance $(U', \bfB)$ of $\bucq(\classU)$ where
  $nr(U') \equiv U$. Such an $U'$ can be found by enumeration of $\classU$ in time that only depends on the parameter.
  Since $nr(U') \equiv U$, the reduction is trivially correct.
\end{proof}

\section{On the Plain Tractability of CSPs}
\label{sec:ptime}

A characterization for the plain (non-parameterized) tractability of CSPs remains an open question. Here we wish to highlight two consequences of our work and recent developments regarding the connection of fractional hypertree width and the Vapnik-Chervonenkis (VC) dimension of a hypergraph presented in~\cite{DBLP:journals/corr/abs-2002-05239}.

\smallskip
\noindent
\textbf{Tractability in natural problem classes.}
Bounded hypertree width ($hw$), generalized hypertree width ($ghw$) and fractional hypertree width ($\fhw$) all represent sufficient conditions for tractable CSP solving, with $\fhw$ being the most general such property we know of. It is known that $hw$ is bounded if and only if $ghw$ is bounded~\cite{DBLP:journals/ejc/AdlerGG07}. Furthermore, there exist classes that exhibit bounded $\fhw$ but unbounded $hw$~\cite{DBLP:journals/talg/GroheM14}.
However,
all known hypergraph classes with bounded $\fhw$ and unbounded $hw$
involve some form of exponential growth that is unlikely to be present in natural problems.
It has remained an open question if this exponential growth is essential for the separation of the two width measures.

Below, we give an answer to this question. The technical details of VC dimension are not important here. Rather we introduce the notion of \emph{\degen} hypergraph classes, a consequence of unbounded VC dimension, to focus on the exponential
character of such classes.
We are able to state that this property is indeed intrinsic to the separation of bounded $\fhw$ and $hw$.
Alternatively, in the contrapositive, we see that for non-\degen classes,
a class has bounded $\fhw$ if and only if it has bounded $hw$. %
In other words, bounded $\fhw$ does not allow for additional tractable cases over bounded $hw$.

\begin{definition}
  Let $\classH$ be a class of hypergraphs. We say that $\classH$ is
  \emph{\degen} if for every integer $n \geq 1$, there exists a
  $H \in \classH$ with a set of $n$ vertices $U \subseteq V(H)$ such
  that $H[U]$ has at least $2^n-1$ distinct edges.
\end{definition}

\begin{theorem}
  \label{thm:degen}
  For any class $\classH$ of hypergraphs, if $\classH$ has unbounded
  hypertree width and bounded fractional hypertree width then $\classH$ is \degen.
\end{theorem}
\begin{proof}[Proof (Sketch)]
  As stated above, exoticness is a
  consequence of unbounded VC dimension, thus we also have that if
  $\classH$ is not \degen, then $\classH$ has bounded VC dimension.

  The key observation is then that the integrality gap for fractional
  edge covers can be bounded by a function of the VC dimension. Hence,
  under bounded VC dimension the integrality gap is constant. By
  applying this observation to every bag of a tree decomposition with
  $\fhw \leq k$ we can see that the tree decomposition will also have
  $ghw$ bounded by some function of $k$ and the VC dimension.  Due to
  space restrictions we refer to the proof Theorem 7.8
  in~\cite{DBLP:journals/corr/abs-2002-05239} for details.
  
  In summary, if $\classH$ has bounded
  VC dimension, then $\classH$ has bounded $\fhw$ iff $\classH$ has bounded $ghw$. Recall from above, that also $ghw(\classH)$ is bounded iff
  $hw(\classH)$ is bounded. Hence, the contrapositive of the
  implication in the theorem holds.
\end{proof}

We can extend the \degen property from hypergraphs to classes of
CSPs in the usual way.
Recall, that in the context of CSPs,
incident edges in the hypergraph correspond to constraints that
involve the variable. Hence, if vertices $U$ have $2^{|U|}-1$ distinct
incident edges in the hypergraph, there exists at least one constraint
for every possible combination of the corresponding variables in the
CSP.  We argue that this situation is highly unnatural and believe that this
motivates further study of the complexity of non-\degen classes of
CSP.

\smallskip
\noindent
\textbf{Semantic width and tractability.}
In the parameterized setting, it is easy to utilize low semantic width to establish upper-bounds as computing the core requires time only in the parameter. For tractability the situation is more problematic. As noted in Section~\ref{sec:csp}, finding the core is intractable. Hence, if we have a class with bounded semantic fractional hypertree width, we know that the problem itself is not difficult, but an efficient solution depends on the hard problem of finding the core. We are caught in an unsatisfactory situation where the origin of the hardness is no longer the actual problem but the concrete formulation.

Part of the issue is that utilizing bounded $\fhw$ for polynomial evaluation requires a concrete decomposition with low $\fhw$, which then guides the efficient solution of the CSP. Without knowing the core we cannot compute the appropriate decomposition.
For bounded generalized hypertree width, Chen and Dalmau were able to show, that for classes of bounded $ghw$, there exists an algorithm for solving  CSPs in polynomial time without requiring the explicit computation of a decomposition~\cite{DBLP:conf/cp/ChenD05}. Their method indeed remains polynomial if only the semantic generalized hypertree width is bounded. Thus, we are able to lift their result to bounded semantic fractional hypertree width for non-\degen classes.

\begin{corollary}
  \label{cor:sfhw}
  Let $\classC$ be a non-\degen class of CSPs with bounded semantic fractional hypertree width. Then $\csp(\classC)$ is tractable.
\end{corollary}

Any more general sufficient property for tractability would likely have to preserve this feature of making use of the width of the core without actually requiring the computation of the core.
Hence, in light of Theorem~\ref{thm:degen} and Corollary~\ref{cor:sfhw} we conclude the section with the following conjecture.

\begin{conjecture}
  \label{conj}
  Let $\classC$ be a class of non-\degen CSPs. Then $\csp(\classC)$ is tractable if and only if $\classC$ has bounded semantic hypertree width.
\end{conjecture}

\section{Conclusion \& Outlook}
\label{sec:conclusion}
We have given characterizations of the fixed-parameter tractable classes of CSPs and UCQs.
This allows us to determine the parameterized complexity of problems that have CSP or UCQ formulations by determining if the class of these formulations has bounded $\ssubw$.
This motivates further work on theoretical tools that help to show whether a class has bounded $\ssubw$. We believe that further study of \emph{adaptive width} \cite{DBLP:journals/mst/Marx11}, which is bounded iff $subw$ is bounded, can be a productive avenue of research here.

The characterization of polynomial time solvable CSPs remains
open. We have motivated a new class of non-\degen problems that merits
further research. In particular, we wish to resolve
Conjecture~\ref{conj}, which we believe to be an important step towards
the general problem. To expand on the ideas from Section~\ref{sec:ptime} we show in Appendix~\ref{sec:shw} that $\semantic{hw}=\sghw$, thus demonstrating that, in contrast to the other investigated widths, $hw$ is not necessarily minimal in the core. 

Recent work has proposed the use of hybrid width parameters for the study
of the computational complexity of CSP, e.g.,~\cite{DBLP:conf/cp/GanianOS19}. Such hybrid width parameters, which consider both the query structure and database content, are a natural avenue for further research.

Moreover, we are intrigued by the connections to VC dimension, which
is an important parameter in learnability theory
\cite{DBLP:journals/jacm/BlumerEHW89}. We plan to further investigate
the nature of the relationship between decomposition methods and learnability theory.

\section*{Acknowledgments}
This work was supported by the Austrian Science Fund (FWF):P30930. Georg Gottlob is a Royal Society Research Professor and acknowledges support by the Royal Society for the present work in the context of the project ``RAISON DATA''  (Project reference: RP\textbackslash{}R1\textbackslash{}201074).
\bibliographystyle{named}
\bibliography{ijcai20}

\appendix

\newcommand{\mononm}{core minimal\xspace}
\newcommand{\xp}[1]{\ensuremath{f^ {-1}({#1})}}

\newcommand{\vc}{\mbox{\rm vc}}
\newcommand{\calO}{{\mathcal O}}
\newcommand{\HH}{\ensuremath{H}}
\newcommand{\cigap}[1]{\mbox{\it cigap}(#1)}
\newcommand{\tigap}[1]{\mbox{\it tigap}(#1)}
\newcommand{\ggnew}[1]{#1}

\section{Full Proofs for Section~\ref{sec:main}}

\begin{definition}
  let $\mathcal{A}$ be the class of all relational structures. We call a function $w\colon \mathcal{A} \to \mathbb{R}^+$ \emph{\mononm} if it is invariant under isomorphisms and for any $\bfA \in \mathcal{A}$: $w(\core(\bfA)) \leq w(\bfA)$.
\end{definition}

 \begin{lemma}
  \label{lem:reform}
  Fix $k \geq 1$, and let $w$ be a \mononm function. For each relational structure $\bfA$ the following are equivalent:
  \begin{enumerate}
  \item There exists a $\bfA'$ homomorphically equivalent to $\bfA$ with $w(\bfA') \leq k$.
  \item $w(\core(\bfA)) \leq k$.
  \end{enumerate}
\end{lemma}

\begin{proof}
  The core of $\bfA$ is always homomorphically equivalent to $\bfA$ and
  therefore the upward implication follows. For the downward
  implication we have $w(\core(\bfA')) \leq w(\bfA')$ by the virtue of $w$ being \mononm. If $\bfA'$ is
  homomorphically equivalent to $\bfA$, then their cores must be
  isomorphic, thus $w(\core(\bfA)) = w(\core(\bfA')) \leq w(\bfA') \leq k$.
\end{proof}

\begin{lemma}
  \label{lem:semcore}
    A function $w$ is \mononm if and only if for all structures $\bfA$ we have that $\semantic{w}(\bfA) = w(\core(\bfA)))$.
\end{lemma}
\begin{proof}
    The implication from left to right is immediate from Lemma~\ref{lem:reform}.
  For the other direction we observe that for any structure $\bfA'$ where $\bfA' \simeq \bfA$ we have $\semantic{w}(\bfA') \leq w(\bfA)$ by definition.
  Thus, from $\bfA \simeq \bfA$ we see $w(\core(\bfA))=\semantic{w}(\bfA) \leq w(\bfA)$.
\end{proof}

A \emph{homomorphism} $G \to H$ for
hypergraphs is a mapping $f \colon V(G) \to V(H)$ s.t. if $e \in E(G)$,
then $\{ f(v) \mid v \in e \} \in E(H)$. Function application is
extended to hyperedges and sets of hyperedges in the usual, element-wise, fashion:  for instance, 
for $e \in E(G)$, we write $f(e)$ to denote $\{ f(v) \mid v \in e \}$.
Likewise, for $E' \subseteq E(G)$, we write $f(E)$ to denote $\{ f(e) \mid e \in E' \}$.
Note that if two structures are homomorphic, then also their associated hypergraphs are homomorphic, while the 
converse is, in general, not true.

\begin{lemma}
  \label{lem:homomorph_cover}
  Let $G$ and $H$ be two hypergraphs and
  let $f$ be a homomorphism from $G$ to $H$. 
  Given a fractional edge cover $\mathbf{x}$ of $G$, define $\mathbf{x'}$ s.t. 
  $$x'_h = \sum_{g \in \xp{h}} x_g \qquad h \in E(H).$$
    Then $\mathbf{x'}$ is a fractional edge cover of $f(V(G))$ with the same total weight as $\mathbf{x}$.
\end{lemma}
\begin{proof}

We will write $I_v$ for the set of all
incident edges of a vertex $v$.  
  We first show that $\mathbf{x'}$ is fractional edge cover.  
 In an initial step we show that for every $E \subseteq E(G)$,
  the $\mathbf{x'}$ weight of edges in $f(E)$ will always be greater or equal to the $\mathbf{x}$ weight of $E$. We will (briefly) abuse notation and write $f^{-1}(f(E))$ when we in fact refer to the \emph{union of all the preimages}, i.e., the set of all the edges that map to edges in $f(E)$. 
  It is then easy to observe $E \subseteq f^{-1}(f(E))$ and, therefore, we also have 
  $$\sum_{h \in f(E)} x'_h = \sum_{h \in f(E)} \sum_{g \in f^{-1}(h)}  x_g 
  \geq \sum_{g \in f^{-1}(f(E))}  x_g   \geq \sum_{g \in E} x_g.$$
    Now, choose an arbitrary  $w \in f(V(G))$ and any $v \in f^{-1}(w)$. In combination with our previous observation we can then conclude:
  $$\sum_{h \in I_w} x'_h \geq \sum_{h \in f(I_v)}  x'_h \geq \sum_{g \in I_v} x_g \geq 1$$
The leftmost inequality holds, because $f(I_v) \subseteq I_w$. 
The rightmost inequality holds, because we are assuming that $\mathbf{x}$ is a fractional edge cover of $G$.
We have thus shown that 
$\mathbf{x'}$ covers $w$. Since  $w \in f(V(G))$ was arbitrarily chosen, we conclude that 
$\mathbf{x'}$ is a fractional edge cover of $f(V(G))$.

To see that the total weights of both covers are the same, observe:
  $$ \sum_{h \in f(E(G))} x'_h = \sum_{h \in f(E(G))} \sum_{g \in \xp{h}} x_g = \sum_{g \in E(G)} x_g$$
The right equality follows from the fact that every edge of $G$ is present in exactly one set $\xp{h}$.
\end{proof}

\begin{lemma}
  \label{lem:corecover}
  The fractional edge cover number $\rho^*$ of a relational structure is \mononm. 
\end{lemma}
\begin{proof}
  Let $G$ be the hypergraph of $\bfA$ and $H$ be the hypergraph of
  $\core(\bfA)$.  Since there is a surjective homomorphism from $\bfA$ to
  $\core(\bfA)$, there exists a surjective homomorphism from $G$ to $H$. Then, by
  Lemma~\ref{lem:homomorph_cover}, for any fractional edge cover of
  $G$ there exists a cover of $H$ with equal weight.
\end{proof}

\begin{lemma}
\label{lem:swidthapp}
The functions $fhw$, $adw$, and $subw$ are \mononm.
\end{lemma}
\begin{proof}
Let $\bfA$ be a relational structure and $f$ an
endomorphism from $\bfA$ to $\core(\bfA)$. W.l.o.g., we may assume $f(v)=v$ for all $v \in f(\bfA)$.  
This can be seen as follows: suppose that $f(v)=v$ does not hold for all $v \in f(\bfA)$.  
Clearly, $f$ restricted to $\core(\bfA)$ must be a variable renaming. Hence, there exists 
the inverse variable renaming $f^{-1} \colon \core(\bfA) \rightarrow \core(\bfA)$. Now set 
$f^* = f^{-1}(f(\cdot))$. Then $f^* \colon \bfA \rightarrow \core(\bfA)$ is the desired endomorphism
from $\bfA$ to $\core(\bfA)$ with $f^*(v)=v$ for all $v \in f^*(\bfA)$.  
  
Let $H=(V(H), E(H))$ denote the hypergraph of $\bfA$ and $H'=(V(H'), E(H'))$ the hypergraph of 
$\core(\bfA) = f(\bfA)$.
Furthermore, let $(T, (B_u)_{u\in V(T)})$ be a tree decomposition of $H$. 
Then we create $(T, (B'_u)_{u \in V(T)})$ with the same structure as the
original decomposition and $B'_u = B_u \cap V(H')$.  This gives us a
tree decomposition of $H'$: for every edge $e \in E(H')$ with 
$e \subseteq B_u$, also $e \subseteq B_u \cap V(H')$ holds, because
$e \subseteq V(H')$.  Removing vertices completely from a
decomposition cannot violate the connectedness condition. Actually, some
bags $B'_u$ might become empty but this is not problematic: either we simply allow empty bags in the 
definition of the various notions of width; or we transform $(T, (B'_u)_{u \in V(T)})$ by
deleting all nodes $u$ with empty bag from $T$ and append every node with a non-empty bag as a (further) child of the nearest ancestor node with non-empty bag.

  \begin{description}[topsep=0pt]
  \item[fhw:] We show that if $(T, (B_u)_{u\in V(T)})$ has
    $\rho_H^*$-width $k$, then $(T, (B'_u)_{u \in V(T)})$ has
    $\rho_{H'}^*$-width $\leq k$: By assumption, there is a fractional
    edge cover $\gamma_u$ of every set $B_u$ with weight $\leq k$.  By
    Lemma~\ref{lem:homomorph_cover}, there exists a cover $\gamma'_u$
    of $f(B_u)$ with weight $\leq k$. What is left to show is that
    $\gamma'_u$ also covers $B'_u$.  Recall, that $f(v)=v$ for any
    $v \in V(H')$ and therefore $f(B_u \cap V(H')) = B_u \cap
    V(H')$. It then becomes easy to see that

    $$ B'_u = B_u \cap V(H') = f(B_u \cap V(H')) \subseteq f(B_u)$$
    and in consequence $\gamma'_u$ clearly also covers $B'_u$.
  \item[subw (and adw):]
    Let $\mathcal{F}$ and $\mathcal{F}'$ be the sets of monotone, edge-dominated, submodular functions on $V(H)$ and $V(H')$ respectively. We show that for every $b' \in \mathcal{F}'$
    there exists $b \in \mathcal{F}$, such that  $b'$-width$(H') \leq b$-width$(H)$:

    \sloppy
    Consider an arbitrary monotone, edge-dominated, submodular function 
    $b' \colon 2^{V(H')}\to \mathbb{R}^+$ with $b'(\emptyset)=0$. This function can be extended to a monotone, edge-dominated, 
    submodular function $b\colon 2^{V(H)} \to \mathbb{R}^+$ on $V(H)$ by setting
    $b(X) = b'(X \cap V(H'))$ for every $X \subseteq V(H)$. 
    Now, for any such $b'$ let  $(T, (B_u)_{u \in V(T)})$ be the tree decomposition  for the original hypergraph with minimal $b$-width $=k$.
    Let $(T, (B'_u)_{u \in V(T)})$ refer to the tree decomposition of the core hypergraph, created by the procedure described above. Clearly
    $(T, (B'_u)_{u \in V(T)})$ has $b'$-width $=k$ because by construction
    $b'(B'_u) = b'(B_u \cap V(H')) = b(B_u)$ for every $u \in
    V(T)$.

    Thus, for every monotone edge-dominated submodular function $b'$ on the core hypergraph $H'$, there exists a function $b$ for $H$ where $b'$-width$(H') \leq b$-width$(H)$. As the submodular width is determined by the supremum over all
    permitted functions we see that $subw(H') \leq subw(H)$.

    For $adw$ observe that the definition of function $b$ and the line of argumentation above 
    still holds if we start off with a monotone, edge-dominated, {\em modular} function $b'\colon 2^{V(H)} \to \mathbb{R}^+$.
    
  \end{description}
\end{proof}

\noindent
\begin{proof}[Proof of Lemma~\ref{lem:swidth}]
The theorem now follows from a straightforward combination of Lemmas~\ref{lem:swidthapp},~\ref{lem:corecover} and~\ref{lem:semcore}.
  
\end{proof}

\section{Full Proofs for Section~\ref{sec:ptime}}

The collapse of bounded fhw and bounded hw for bounded VC-dimension is implicitly shown in the proof of Theorem~7.8 of \cite{DBLP:journals/corr/abs-2002-05239}. The statement there puts an emphasis on the computational complexity of fhw checking and does not explicitly state the collapse.
For the sake of completeness and for ease of reading
we restate the theorem in a way that fits our setting and repeat the relevant definitions and segment of the proof here.

\begin{definition}[\cite{1972sauer,1971vc}]
\label{def:vc}
Let $\HH=(V(H),E(H))$ be a hypergraph and $X\subseteq V(H)$ a set of vertices. 
Denote by $E(H)|_X$ the set 
$E(H)|_X =\{X \cap e\, |\, e\in E(H)\}$. 
The vertex set $X$ is called {\em shattered} if 
$E(H)|_X=2^X$.
The {\em Vapnik-Chervonenkis dimension (VC dimension) $\vc(\HH)$} of $\HH$ is 
the maximum cardinality of a shattered subset of $V(H)$. 
\end{definition}

\begin{definition}
\label{def:transversality}
Let $H = (V(H),E(H))$ be a hypergraph. A {\em transversal}  (also known as {\em 
hitting set\/})
of $H$ is a subset $S \subseteq V(H)$ that has a non-empty intersection with 
every edge of $H$. 
The {\em transversality} $\tau(H)$  of $H$  is the 
minimum cardinality of all transversals of $H$.

Clearly, $\tau(H)$ corresponds to the minimum of the following integer linear 
program: 
find a mapping $w: V\rightarrow \{0,1\}$ 
which minimizes $\Sigma_{v\in V(H)}w(v)$ under the condition that
$\Sigma_{v\in e}w(v)\geq 1$ holds for each hyperedge $e\in E$.

The {\em fractional transversality}  $\tau^*$ of $H$ is defined as the minimum 
of
the above linear program when dropping the integrality condition,
thus allowing mappings $w: V\rightarrow \mathbb{R}_{\geq 0}$.
Finally, the {\em transversal integrality gap} $\tigap{\HH}$ of $\HH$ is the 
ratio $\tau(\HH)/\tau^*(\HH)$.
\end{definition}

Recall that computing the mapping $\lambda_u$ for some node $u$ in a
GHD can be seen as searching for a minimal edge cover $\rho$ of the
vertex set $B_u$, whereas computing
$\gamma_u$ in an FHD corresponds to the search for a minimal
fractional edge cover $\rho^*$ \cite{2014grohemarx}. Again, these
problems can be cast as linear programs where the first problem has
the integrality condition and the second one has not.  Further, we can
define the {\em cover integrality gap} $\cigap{H}$ of $H$ as the ratio
$\rho(\HH)/\rho^*(\HH)$.

\begin{lemma}\label{theo:ApproxVC}
  Let $\classH$ be a class of hypergraphs with VC-dimension bounded by
  some constant $d$.  Then for every hypergraph $H \in \classH$
we have
  \(
     hw(H) = O(fhw(H) \log fhw(H))
  \).
\end{lemma}

\begin{proof}
The proof proceeds in several steps.

\smallskip

\noindent
{\em Reduced hypergraphs.} 
We consider, w.l.o.g., only
hypergraphs $H$ that satisfy the following 4 conditions: (1) $H$ has no  isolated vertices 
and (2) no empty edges. Moreover, (3) no two distinct vertices in $H$  
have the same edge-type
(i.e., the two vertices occur in precisely the same edges)
and 
(4) no two distinct edges in $H$ have the same vertex-type
(i.e., we exclude duplicate edges). 
Hypergraphs satisfying these conditions will be called ``reduced''. For a full discussion on why these assumptions can be made without loss of generality we refer to~\cite{DBLP:journals/corr/abs-2002-05239}.

\smallskip

\noindent
{\em Dual hypergraphs.} Given a hypergraph $H = \{V,E)$, the dual hypergraph 
$H^d  = (W,F)$ 
is defined as $W = E$ and $F = \{ \{e \in E \mid v \in e\} \mid v \in V\}$.
For the rest of this proof we consider only reduced hypergraphs. 
This ensures that 
$(H^d)^d = H$  holds.

It is well-known and easy to verify that the following relationships
between $H$ and $H^d$ hold for any reduced hypergraph $H$,  
(see, e.g., \cite{duchet1996hypergraphs}):

\smallskip

(1) The edge coverings of $\HH$ and the transversals of $\HH^d$
coincide.

(2) The fractional %
edge coverings of $\HH$ and the fractional transversals of 
$\HH^d$
coincide.

(3) $\rho(\HH)=\tau(\HH^d)$, $\rho^*(\HH)=\tau^*(\HH^d)$, and
$\cigap{\HH}=\tigap{\HH^d}$.

\smallskip

\noindent
{\em VC-dimension.} By a classical result (\cite{ding1994} Theorem (5.4), see also \cite{Bronnimann1995} for related results), for every 
hypergraph $H = (V(H),E(H))$  with at least two edges we have:
$$\tigap{H}= \tau(H)/\tau^*(H) \leq 2\vc(H)\log(11\tau^*(H)).$$
For hypergraphs $H$ with a single edge only, $\vc(H)=0$, and thus the above inequation does not hold. However, for such hypergraphs 
$\tau(H)=\tau^*(H)=1$. By putting this together, we get:
$$\tigap{H}= \tau(H)/\tau^*(H) \leq \max(1,2\vc(H)\log(11\tau^*(H))).$$
Moreover, in \cite{assouad1983}, it is shown that $\vc(\HH^d)<2^{\vc(\HH)+1}$ 
always
holds.
In total, we thus get 

\begin{align*}
\cigap{\HH}=\tigap{\HH^d} & \leq \ggnew{\max(1,}2\vc(\HH^d)\log(11\tau^*(\HH^d))\ggnew{)} \\
& \leq \ggnew{\max(1,}2^{\vc(\HH)+2}\log(11\rho^*(\HH))\ggnew{)}\\
& \ggnew{\leq \max(1,2^{d+2}\log(11\rho^*(\HH)))},\ 
                                                                     \ggnew \\
                          &{\mbox{which is\ } O(\log\rho^*(H))}.
\end{align*}

\noindent
Suppose that $H$ has an FHD $\left< T, (B_u)_{u\in V(T)}, (\lambda)_{u\in V(T)} 
\right>$
of width $k$. Then there exists a GHD of $H$ of width %
$O(k \cdot \log k)$. Indeed, we can find such a GHD by leaving the 
tree 
structure $T$ and the bags $B_u$ for every node $u$ in $T$ unchanged and 
replacing each fractional edge cover $\gamma_u$ of $B_u$ by an optimal integral 
edge cover $\lambda_u$ of $B_u$. By the above inequality, we thus increase the 
weight at each node $u$ only by a factor $\calO(\log k)$. Moreover, we know 
from 
\cite{DBLP:journals/ejc/AdlerGG07} that $hw(H) \leq 3 \cdot ghw(H) + 1$ holds.
In other words, there also exists an HD of $H$ whose width is $O(k \cdot\log k)$. In particular, this also applies to the minimal width FHD, concluding the proof.
\end{proof}

\begin{lemma}
  \label{lem:vcexotic}
  Let $\classH$ be a hypergraph class. If $\classH$ has unbounded VC dimension, then $\classH$ is exotic.
\end{lemma}
\begin{proof}
  Assuming that $\classH$ has unbounded VC dimension, we show for every integer $n \geq 1$ that there exists a hypergraph $H \in \classH$ with a set of vertices $U \subseteq V(H)$ such that $H[U]$ has at least $2^n-1$ distinct edges.

  Suppose some fixed $n \geq 1$ and let $H \in \classH$ be a
  hypergraph with VC dimension at least $n$. Since $\classH$ has
  unbounded VC dimension such a $H$ always exists. By definition $H$
  now has a shattered set of vertices $X$ with $|X| \geq n$. From the
  similarity in the definition of shattered subsets and vertex induced
  hypergraphs we can observe $H[X] = (X, E(H)|_X \setminus \emptyset)$. Now since $E(H)|_X = 2^X$
  it consists of at least $2^n$ distinct edges. As we remove only one (the empty set), we see that the statement holds.
\end{proof}

\begin{proof}[Proof of Theorem~\ref{thm:degen}] By contraposition of Lemma
\ref{lem:vcexotic} we have that non-exotic classes of hypergraphs also
have bounded VC dimension. From Lemma~\ref{theo:ApproxVC} we see that
bounded fractional hypertree width implies bounded hypertree
width. Lastly, a hypertree decomposition is a special case of a
fractional hypertree decomposition. Hence, bounded hypertree width
also implies bounded fractional hypertree width. 
\end{proof}

\mathchardef\mhyphen="2D
\newcommand{\chw}{\ensuremath{\core\mhyphen\mathit{hw}}}
\newcommand{\hdecomp}{\ensuremath{\left< T, (B_u)_{u\in T}, (\lambda_u)_{u \in T} \right>}}
\newcommand{\ghw}{\mathit{ghw}}
\newcommand{\hw}{\mathit{hw}}
\section{Semantic Hypertree Width}
\label{sec:shw}

We start with some additional definitions that are necessary for this section.

A \emph{hypertree decomposition}~\cite{DBLP:journals/jcss/GottlobLS02} of a hypergraph $H$ is a tuple $\hdecomp$, where $T$ is a rooted tree, for every node $u$ of the tree, $B_u \subseteq V(H)$ is called the \emph{bag} of node $u$, and $\lambda_u \subseteq E(H)$ is the \emph{cover} of $u$.
Furthermore, $\hdecomp$ must satisfy the following properties.
\begin{enumerate}
\item The subgraph $T_v  =\{u \in T \mid v \in B_u\}$ for vertex $v \in V(H)$ is a tree.
\item For every $e \in E(H)$ there exists a $u \in T$ such that $e \subseteq B_u$.
\item For every node $u$ in $T$ it holds that $B_u \subseteq \bigcup \lambda_u$.
\item Let $T_u$ be the subtree of $T$ rooted at node $u$ and let $B(T_u)$ be the union of all bags of nodes in $T_u$.
  For every node $u$ in $T$ it holds that $\bigcup \lambda_u \cap B(T_u) \subseteq B_u$.
\end{enumerate}
The first property is commonly referred to as the \emph{connectedness condition} and the fourth property is called the \emph{special condition}. The \emph{hypertree width} ($\hw$) of a hypertree decomposition is $\max_{u\in T}(|\lambda_u|)$ and the hypertree width of $H$ ($\hw(H)$) is the minimal width of all hypertree decompositions of $H$.

If we exclude the special condition in the above list of properties, we obtain the definition of a \emph{generalized hypertree decomposition} (GHD). The \emph{generalized hypertree width} of hypergraph $H$ ($\ghw(H)$) is defined analogously to before as the minimal width of all generalized hypertree decompositions of $H$. This definition is equivalent to definition of $\ghw$ given in Section~\ref{sec:prelim}.

The special condition demands that if a vertex $v$ occurs in an edge $e$ in $\lambda_u$ and in a bag in the subtree below $u$, then $v$ must also appear in $B_u$. If this property is violated in a GHD, we say that $e$ causes a \emph{special condition violation} (SCV) at node $u$.

\subsection{Semantic Hypertree Width is Semantic Generalized Hypertree Width}

The approach used for
$\fhw, adw, subw$ in the previous section does not work for hypertree with. Constructing a new
tree decomposition by intersecting the bags with the vertices in $H'$
can break the special condition.
Indeed, we will show that $\semantic{hw}(\bfA) = \sghw(\bfA)$ for all structures $\bfA$. Our argument is based on a construction of equivalent structures that fix special condition violations in a generalized hypertree decomposition. This observation positions hypertree width uniquely against all other widths studied in the previous section.

\begin{lemma}
  \label{lem:shwmagic}
  Let $\bfA$ be a structure with $ghw(\bfA) \leq k$. Then, there
  exists an $\bfA'$ with $\bfA' \simeq \bfA$ and $hw(\bfA') \leq k$.
\end{lemma}
\begin{proof}
  Let $\mathcal{D} = \left< T, (B_u)_{u\in T}, (\lambda_u)_{u \in T} \right>$ be a
  GHD of $H(\bfA)$ with width $k$. We will show how to add a new tuple to $\bfA$ to get a new $\bfA'$ such that $\bfA'\simeq \bfA$
  and  $H(\bfA')$ has a GHD with width $k$ with \emph{fewer} special condition violations (SCVs) than $\mathcal{D}$. Iterating this step will ultimately lead to a structure that is equivalent to $\bfA$ and has a GHD of width $k$ with no SCVs, i.e., an HD.

  Let $u$ be any node in $T$ where the special condition is violated, i.e., there is some $e^* \in \lambda_u$ such that $(e^* \cap B(T_u)) \not \subseteq (e^* \cap B_u)$. Let $v_1, \dots, v_\ell$ be the vertices in $e^*$ that are not in $B_u$ and let $R^\bfA$ be the relation in $\bfA$ that contains a tuple $t$ that becomes the edge $e^*$ in the hypergraph\footnote{There can be multiple tuples in different relations that correspond to the edge $e^*$ in $H(\bfA)$, it does not matter to which we apply the procedure.}. Now, for every $i\in[\ell]$, create a fresh constant $x_i$. We create a new tuple $t'$ from $t$ by replacing every $v_i$ (the vertices that witness the SCV) by the fresh $x_i$. All other constants in $t$ are copied with no change to $t'$. Add $t'$ to $R$ to obtain the new structure $\bfA'$.

  We first verify that $\bfA'$ has a GHD $\mathcal{D'}$  with width $k$ and one less SCV than $\mathcal{D}$. We can find such a $\mathcal{D'}$ by simply copying $\mathcal{D}$ and only updating $\gamma_u$ and $B_u$ as follows. In $\gamma_u$ we replace $e^*$ by the edge $e'$ that corresponds to our newly created tuple $t'$. To $B_u$ we add all the newly created $x_1, \dots, x_\ell$.  With this change to the bag, every edge remains covered (and the new $e'$ is now covered by $B_u$) and connectedness is unaffected (the new vertices $x_1,\dots,x_\ell$ occur only in this bag). Furthermore, it is clear that we still have $B_u \subseteq \bigcup \lambda_u$ after this update. Finally, while $e$ caused an SCV at node $u$, this is no longer the case with $e'$ since $e'$ is fully contained in $B_u$. Thus our new $\mathcal{D'}$ is a valid GHD of $H(\bfA')$ with one less SCV than $\mathcal{D}$.

  To finalize our argument we still need to show that $\bfA \simeq \bfA'$. First, since every relation in $\bfA'$ is a superset of a relation in $\bfA$, the identity function is a homomorphism from $\bfA$ to $\bfA'$. For the other direction, consider the function $f\colon A' \to A$  that maps $x_i \mapsto v_i$ for $i \in [\ell]$ and every other constant in $\bfA'$ to itself.
  Clearly, $f(t')=t \in R^\bfA$ for the tuples from our construction above. For all other tuples $f$ is the identity function since they do not contain any of the fresh constants $x_i$. All those tuples are present in both structures (in the same relations). Thus, $f$ is a homomorphism from $\bfA'$ to $\bfA$.

  We can therefore move along equivalent structures to (strictly) monotonically decrease the number of SCVs, ultimately yielding an HD with width $k$.
\end{proof}

\begin{theorem}
  \label{thm:shw}
  For any relational structure $\bfA$ it holds that $\semantic{hw}(\bfA) = \sghw(\bfA)$.
\end{theorem}
\begin{proof}
  Suppose $\sghw(\bfA) = k$, then $ghw(\core(\bfA)) = k$  by Lemma~\ref{lem:swidth}. From Lemma~\ref{lem:shwmagic} it now follows that there exists an $\bfA'$ such that $hw(\bfA')\leq k$ and $\bfA' \simeq \core(\bfA) \simeq \bfA$. Thus, $\semantic{hw}(\bfA) \leq k$. On the other hand, in general for any hypergraph $H$ we have $ghw(H)\leq hw(H)$ and therefore $\semantic{hw}(\bfA)$ can not be lower than $\sghw(\bfA)$. Hence, $\semantic{hw}(\bfA)=k=\sghw(\bfA)$.
\end{proof}

 \end{document}